\documentclass[11pt,a4paper]{amsart}

\usepackage[T1]{fontenc}
\usepackage[utf8]{inputenc}
\usepackage[english]{babel}

\usepackage[a4paper,margin=2.8cm]{geometry}

\usepackage{amsmath,amssymb,amsfonts,amsthm,mathtools}
\usepackage{mathrsfs}
\usepackage{bm}
\usepackage{tikz-cd}
\usepackage{tikz-cd}
\usepackage{enumitem}
\usepackage{xcolor}

\usepackage[colorlinks=true,linkcolor=blue,citecolor=blue,urlcolor=blue]{hyperref}

\theoremstyle{plain}
\newtheorem{theorem}{Theorem}[section]
\newtheorem{proposition}[theorem]{Proposition}
\newtheorem{lemma}[theorem]{Lemma}
\newtheorem{corollary}[theorem]{Corollary}

\theoremstyle{definition}
\newtheorem{definition}[theorem]{Definition}
\newtheorem{remark}[theorem]{Remark}

\newcommand{\C}{\mathbb{C}}
\newcommand{\R}{\mathbb{R}}
\newcommand{\CP}{\mathbb{CP}}

\newcommand{\PC}{\mathrm{PC}}
\newcommand{\Tr}{\mathrm{Tr}}
\newcommand{\rank}{\mathrm{rank}}

\newcommand{\Barg}{\mathrm{B}}
\newcommand{\Sphere}{\mathbb{S}}

\newcommand{\ket}[1]{\left|#1\right\rangle}
\newcommand{\bra}[1]{\left\langle #1\right|}
\newcommand{\ip}[2]{\left\langle #1\,\middle|\,#2\right\rangle}

\title[Pairwise comparisons and qubit states]{Remarks on pairwise comparisons, transition amplitudes, and qubit states}

\author[Jean-Pierre Magnot]{Jean-Pierre Magnot}
\address{{SFR MATHSTIC, LAREMA, Universit\'e d’Angers, 2 Bd Lavoisier, 
49045 Angers cedex 1, France;  Lyc\'ee Jeanne d'Arc, 40 avenue de Grande Bretagne, 63000 Clermont-Ferrand, 
France}; Lepage Research Institute, 17 novembra 1, 081 16 Presov, Slovakia}
\email{\small magnot@math.cnrs.fr; jean-pierr.magnot@ac-clermont.fr}


\begin{document}

\begin{abstract}
We discuss a pairwise-comparison viewpoint on finite families of qubit states. 
Starting from transition amplitudes between pure states, we distinguish three associated levels of comparison data: complex amplitudes, transition probabilities, and phase-valued pairwise comparisons. 
In the non-orthogonal case, the phase data define a \(U(1)\)-valued reciprocal pairwise comparison structure. 
We show that the corresponding triangular defects are naturally related to normalized Bargmann invariants and therefore to geometric phases. 
This gives a simple interpretation of inconsistency-type quantities in terms of quantum kinematics. 
We also comment on realizability constraints coming from Gram matrices of rank at most two, and on the passage from unitary phase data to more general transition data. 
The aim of the paper is mainly conceptual: to isolate a common language between pairwise comparisons and elementary quantum geometry.
\end{abstract}

\maketitle

\section{Introduction}

Pairwise comparisons are usually introduced as tools for ranking, decision, or relative evaluation.
However, they also provide a simple mathematical language for relational data: what is given is not an absolute value attached to an object, but a comparison between two objects.
This point of view becomes particularly suggestive in quantum theory, where transition amplitudes and transition probabilities are themselves relational quantities attached to pairs of states rather than to isolated states.

The purpose of this note is to examine this analogy in the elementary case of qubit states.
Given a finite family of pure states
\[
\psi_1,\dots,\psi_N \in \C^2,
\qquad
\|\psi_i\|=1,
\]
one may consider the transition amplitudes
\[
g_{ij}=\langle \psi_i,\psi_j\rangle,
\]
the transition probabilities
\[
p_{ij}=|g_{ij}|^2,
\]
and, whenever \(g_{ij}\neq 0\), the normalized phases
\[
u_{ij}=\frac{g_{ij}}{|g_{ij}|}\in U(1).
\]
These three levels of data define different kinds of pairwise comparison structures.
The first retains the full complex transition information, the second forgets phase, and the third keeps only the phase part.

The phase-valued level is especially close to the usual formalism of reciprocal pairwise comparisons.
Indeed, if all overlaps are nonzero, then
\[
u_{ii}=1,
\qquad
u_{ji}=u_{ij}^{-1},
\]
so that the matrix \(U=(u_{ij})\) defines a \(U(1)\)-valued reciprocal pairwise comparison matrix.
Its triangular defects
\[
\kappa_{ijk}=u_{ij}u_{jk}u_{ki}
\]
then measure a failure of multiplicative coherence.
In the present quantum setting, however, these defects are not merely algebraic artifacts: they are naturally related to normalized Bargmann invariants and therefore to geometric phase phenomena.

This observation suggests that inconsistency-type quantities in pairwise comparison theory admit a natural interpretation in elementary quantum geometry.
Conversely, it shows that transition data for qubit states can be organized in a language familiar from pairwise comparison theory.
The relation is not exact at every level: not every reciprocal \(U(1)\)-valued matrix comes from a family of qubit states, and the full amplitude data are subject to strong positivity and rank constraints through the associated Gram matrix.
Nevertheless, the comparison is sufficiently rigid to be mathematically meaningful.

Our aim here is modest.
We do not propose a full reconstruction theory, nor a general theory of quantum pairwise comparisons.
Rather, we isolate a simple dictionary between reciprocal comparison data, transition amplitudes, and geometric phases in the qubit case.
The paper is organized as follows.
In Section~2 we introduce the relevant comparison structures associated with qubit transition data.
Section~3 is devoted to phase-valued comparisons and triangular defects.
In Section~4 we relate these defects to Bargmann invariants and geometric phases.
Section~5 contains remarks on realizability and Gram matrix constraints.
We conclude with a few comments on possible extensions..

\section{Pairwise comparisons and qubit transition data}

\subsection{Pairwise comparison matrices}

We begin by recalling the standard notion of a pairwise comparison matrix.

\begin{definition}
A \emph{multiplicative pairwise comparison matrix} of order \(N\) is a matrix
\[
A=(a_{ij})_{1\leq i,j\leq N}
\]
with entries in \(\R_{+}^{\times}\) such that
\[
a_{ii}=1
\qquad\text{and}\qquad
a_{ji}=a_{ij}^{-1}
\]
for all \(i,j\).
\end{definition}

Thus the entry \(a_{ij}\) measures the relative comparison of item \(i\) with item \(j\), while reciprocity imposes that reversing the order inverts the comparison.

A classical coherence condition is the multiplicative transitivity relation
\[
a_{ij}a_{jk}=a_{ik}
\qquad
\text{for all } i,j,k.
\]
Equivalently, the triangular defect
\[
a_{ij}a_{jk}a_{ki}
\]
should be equal to \(1\).

For the purposes of the present note, it is convenient to extend this notion to group-valued coefficients.

\begin{definition}
Let \(G\) be a group, written multiplicatively, with neutral element \(e\).
A \emph{\(G\)-valued reciprocal pairwise comparison matrix} of order \(N\) is a matrix
\[
U=(u_{ij})_{1\leq i,j\leq N}
\]
with entries in \(G\) such that
\[
u_{ii}=e
\qquad\text{and}\qquad
u_{ji}=u_{ij}^{-1}
\]
for all \(i,j\).

We denote by \(\PC_N(G)\) the set of such matrices.
\end{definition}

When \(G=U(1)\), this gives the natural notion of a phase-valued reciprocal comparison matrix.
For such a matrix, the basic defect attached to a triple \((i,j,k)\) is
\[
\kappa_{ijk}:=u_{ij}u_{jk}u_{ki}\in G.
\]
If \(\kappa_{ijk}=e\), the triple is multiplicatively coherent.

\subsection{Qubit states}

We now recall the elementary notion of a qubit state.

\begin{definition}
A \emph{pure qubit state} is a one-dimensional subspace of \(\C^2\), that is, a point of the complex projective line
\[
\CP^1.
\]
Equivalently, it may be represented by a normalized vector
\[
\psi\in \C^2,
\qquad
\|\psi\|=1,
\]
defined up to multiplication by a phase factor \(e^{i\theta}\in U(1)\).
\end{definition}

Thus two normalized vectors \(\psi,\phi\in\C^2\) represent the same pure qubit state if and only if
\[
\phi=e^{i\theta}\psi
\]
for some \(\theta\in\R\).

\begin{remark}
The space of pure qubit states is naturally identified with the Bloch sphere.
We shall not use this identification immediately, but it will be useful later when geometric phase interpretations are discussed.
\end{remark}

In what follows, we consider a finite family of normalized vectors
\[
\psi_1,\dots,\psi_N\in \C^2,
\qquad
\|\psi_i\|=1.
\]
Strictly speaking, the physical states are the rays \([\psi_i]\in \CP^1\), but it is convenient to choose normalized representatives.

\subsection{Transition amplitudes, probabilities, and phases}

To such a family \((\psi_i)_{1\leq i\leq N}\), one associates the matrix of transition amplitudes
\[
G=(g_{ij})_{1\leq i,j\leq N},
\qquad
g_{ij}:=\ip{\psi_i}{\psi_j}.
\]

\begin{definition}
The matrix
\[
G=(g_{ij})_{1\leq i,j\leq N},
\qquad
g_{ij}=\ip{\psi_i}{\psi_j},
\]
is called the \emph{Gram matrix} of the family \((\psi_i)\).
\end{definition}

Its basic properties are classical:

\begin{proposition}
The Gram matrix \(G\) satisfies:
\begin{enumerate}[label=\textup{(\roman*)}]
\item \(g_{ii}=1\) for all \(i\);
\item \(g_{ji}=\overline{g_{ij}}\) for all \(i,j\);
\item \(G\) is positive semidefinite;
\item \(\rank(G)\leq 2\).
\end{enumerate}
\end{proposition}

\begin{proof}
The first two statements follow directly from the definition of the Hermitian scalar product and the normalization \(\|\psi_i\|=1\).
For positivity, let \(c_1,\dots,c_N\in\C\). Then
\[
\sum_{i,j=1}^N \overline{c_i}\, g_{ij}\, c_j
=
\sum_{i,j=1}^N \overline{c_i}\,\ip{\psi_i}{\psi_j}\,c_j
=
\left\|\sum_{j=1}^N c_j \psi_j\right\|^2
\geq 0.
\]
Hence \(G\) is positive semidefinite.
Finally, since the vectors \(\psi_1,\dots,\psi_N\) lie in \(\C^2\), the span of the family has dimension at most \(2\), so \(\rank(G)\leq 2\).
\end{proof}

From \(G\), one obtains two simpler levels of pairwise data.

\begin{definition}
The \emph{transition probability matrix} associated with \((\psi_i)\) is
\[
P=(p_{ij})_{1\leq i,j\leq N},
\qquad
p_{ij}:=|g_{ij}|^2=|\ip{\psi_i}{\psi_j}|^2.
\]
\end{definition}

The coefficients \(p_{ij}\) satisfy
\[
p_{ii}=1,
\qquad
p_{ij}=p_{ji},
\qquad
0\leq p_{ij}\leq 1.
\]
Thus \(P\) is a symmetric matrix of transition probabilities rather than a reciprocal pairwise comparison matrix in the usual multiplicative sense.

If, in addition, all overlaps are nonzero, one may also extract their phases.

\begin{definition}
Assume that
\[
\ip{\psi_i}{\psi_j}\neq 0
\qquad
\text{for all } i,j.
\]
The associated \emph{phase comparison matrix} is
\[
U=(u_{ij})_{1\leq i,j\leq N},
\qquad
u_{ij}:=\frac{\ip{\psi_i}{\psi_j}}{|\ip{\psi_i}{\psi_j}|}\in U(1).
\]
\end{definition}

\begin{proposition}
Under the non-orthogonality assumption
\[
\ip{\psi_i}{\psi_j}\neq 0
\qquad
\text{for all } i,j,
\]
the matrix \(U=(u_{ij})\) belongs to \(\PC_N(U(1))\).
\end{proposition}

\begin{proof}
For every \(i\),
\[
u_{ii}=\frac{\ip{\psi_i}{\psi_i}}{|\ip{\psi_i}{\psi_i}|}
=\frac{1}{1}=1.
\]
Moreover,
\[
u_{ji}
=
\frac{\ip{\psi_j}{\psi_i}}{|\ip{\psi_j}{\psi_i}|}
=
\frac{\overline{\ip{\psi_i}{\psi_j}}}{|\ip{\psi_i}{\psi_j}|}
=
u_{ij}^{-1},
\]
since \(u_{ij}\in U(1)\).
Hence \(U\in \PC_N(U(1))\).
\end{proof}

\subsection{Three levels of pairwise quantum data}

The family \((\psi_i)\) therefore gives rise to three distinct levels of pairwise information:
\[
g_{ij}=\ip{\psi_i}{\psi_j},
\qquad
p_{ij}=|g_{ij}|^2,
\qquad
u_{ij}=\frac{g_{ij}}{|g_{ij}|}
\quad \text{when } g_{ij}\neq 0.
\]

These three levels should be clearly distinguished.

\begin{itemize}
\item The amplitudes \(g_{ij}\) retain the full complex transition data.
\item The probabilities \(p_{ij}\) retain only the moduli.
\item The phases \(u_{ij}\) retain only the phase information and define a \(U(1)\)-valued reciprocal comparison structure.
\end{itemize}

The main point of this note is that the last level is close to group-valued pairwise comparison theory, while the first one keeps track of additional non-unitary information through the moduli.
The interplay between these two levels will lead naturally to triangular defects, Bargmann invariants, and geometric phases.

\begin{remark}
Not every matrix in \(\PC_N(U(1))\) arises from a family of qubit states.
Indeed, realizability requires the existence of a positive semidefinite Hermitian Gram matrix of rank at most \(2\) whose normalized phase part is the prescribed matrix.
We shall return to this point later.
\end{remark}

\section{Triangular defects and Bargmann invariants}

\subsection{Triangular defects of phase comparisons}

Let
\[
U=(u_{ij})_{1\leq i,j\leq N}\in \PC_N(U(1))
\]
be the phase comparison matrix associated with a non-orthogonal family of normalized qubit representatives
\[
\psi_1,\dots,\psi_N\in \C^2.
\]
A first natural quantity attached to a triple of indices is the corresponding triangular defect.

\begin{definition}
Let \(i,j,k\in\{1,\dots,N\}\) be pairwise distinct.
The \emph{triangular defect} associated with the triple \((i,j,k)\) is
\[
\kappa_{ijk}:=u_{ij}u_{jk}u_{ki}\in U(1).
\]
\end{definition}

This quantity measures the failure of multiplicative coherence of the phase comparison matrix on the cycle \(i\to j\to k\to i\).
Indeed, if one had
\[
u_{ij}u_{jk}=u_{ik},
\]
then automatically
\[
\kappa_{ijk}=u_{ij}u_{jk}u_{ki}=u_{ik}u_{ki}=1.
\]

\begin{remark}
In the classical theory of multiplicative pairwise comparisons, the relation
\[
a_{ij}a_{jk}a_{ki}=1
\]
expresses consistency on the triple \((i,j,k)\).
The quantity \(\kappa_{ijk}\) is therefore the natural \(U(1)\)-valued analogue of a local inconsistency factor.
\end{remark}

The elementary symmetries of these defects follow directly from reciprocity.

\begin{lemma}
For every pairwise distinct \(i,j,k\),
\[
\kappa_{ikj}=\kappa_{ijk}^{-1}.
\]
More generally, cyclic permutations leave \(\kappa_{ijk}\) unchanged, while transpositions invert it.
\end{lemma}

\begin{proof}
Since multiplication in \(U(1)\) is commutative,
\[
\kappa_{jki}=u_{jk}u_{ki}u_{ij}=\kappa_{ijk},
\qquad
\kappa_{kij}=u_{ki}u_{ij}u_{jk}=\kappa_{ijk}.
\]
On the other hand,
\[
\kappa_{ikj}=u_{ik}u_{kj}u_{ji}
=u_{ki}^{-1}u_{jk}^{-1}u_{ij}^{-1}
=(u_{ij}u_{jk}u_{ki})^{-1}
=\kappa_{ijk}^{-1}.
\qedhere
\]
\end{proof}

Thus the argument of \(\kappa_{ijk}\) changes sign under orientation reversal of the triangle.

\subsection{Bargmann invariants}

We now recall the elementary form of the Bargmann invariant associated with a triple of pure states.

\begin{definition}
Let \(\psi_i,\psi_j,\psi_k\in \C^2\) be normalized vectors.
The associated \emph{Bargmann invariant} is
\[
\Barg_{ijk}
:=
\ip{\psi_i}{\psi_j}\ip{\psi_j}{\psi_k}\ip{\psi_k}{\psi_i}
=
g_{ij}g_{jk}g_{ki}\in \C.
\]
\end{definition}

\begin{remark}
The quantity \(\Barg_{ijk}\) depends on the chosen normalized representatives \(\psi_i,\psi_j,\psi_k\), but only through a global phase-invariant combination.
In particular, it is unchanged when each \(\psi_\ell\) is multiplied by an arbitrary phase factor \(e^{i\theta_\ell}\).
\end{remark}

\begin{lemma}
The Bargmann invariant \(\Barg_{ijk}\) is invariant under the rephasing
\[
\psi_\ell \longmapsto e^{i\theta_\ell}\psi_\ell,
\qquad
\theta_\ell\in\R,
\]
of the representatives of the rays \([\psi_\ell]\in \CP^1\).
\end{lemma}

\begin{proof}
Under such a rephasing,
\[
\ip{\psi_i}{\psi_j}\longmapsto e^{-i\theta_i}e^{i\theta_j}\ip{\psi_i}{\psi_j},
\]
\[
\ip{\psi_j}{\psi_k}\longmapsto e^{-i\theta_j}e^{i\theta_k}\ip{\psi_j}{\psi_k},
\]
\[
\ip{\psi_k}{\psi_i}\longmapsto e^{-i\theta_k}e^{i\theta_i}\ip{\psi_k}{\psi_i}.
\]
Multiplying these three factors, all phases cancel:
\[
e^{-i\theta_i}e^{i\theta_j}
\cdot
e^{-i\theta_j}e^{i\theta_k}
\cdot
e^{-i\theta_k}e^{i\theta_i}
=1.
\]
Hence \(\Barg_{ijk}\) is unchanged.
\end{proof}

This is one of the basic reasons why \(\Barg_{ijk}\) is geometrically meaningful: it depends only on the rays, not on the chosen vectors representing them.

\subsection{Normalized Bargmann invariants}

When all pairwise overlaps involved in a triple are nonzero, the Bargmann invariant can be normalized to a phase.

\begin{definition}
Assume that
\[
\ip{\psi_i}{\psi_j}\neq 0,\qquad
\ip{\psi_j}{\psi_k}\neq 0,\qquad
\ip{\psi_k}{\psi_i}\neq 0.
\]
The corresponding \emph{normalized Bargmann invariant} is
\[
\widetilde{\Barg}_{ijk}
:=
\frac{\Barg_{ijk}}{|\Barg_{ijk}|}
=
\frac{\ip{\psi_i}{\psi_j}\ip{\psi_j}{\psi_k}\ip{\psi_k}{\psi_i}}
{|\ip{\psi_i}{\psi_j}\ip{\psi_j}{\psi_k}\ip{\psi_k}{\psi_i}|}
\in U(1).
\]
\end{definition}

Since
\[
|\Barg_{ijk}|
=
|g_{ij}g_{jk}g_{ki}|
=
|g_{ij}|\;|g_{jk}|\;|g_{ki}|,
\]
we immediately obtain the following identification.

\begin{proposition}\label{prop:kappa-bargmann}
For every triple \((i,j,k)\) such that all three overlaps are nonzero,
\[
\kappa_{ijk}
=
u_{ij}u_{jk}u_{ki}
=
\widetilde{\Barg}_{ijk}.
\]
\end{proposition}

\begin{proof}
By definition,
\[
u_{ij}=\frac{g_{ij}}{|g_{ij}|},
\qquad
u_{jk}=\frac{g_{jk}}{|g_{jk}|},
\qquad
u_{ki}=\frac{g_{ki}}{|g_{ki}|}.
\]
Therefore
\[
u_{ij}u_{jk}u_{ki}
=
\frac{g_{ij}g_{jk}g_{ki}}{|g_{ij}||g_{jk}||g_{ki}|}
=
\frac{\Barg_{ijk}}{|\Barg_{ijk}|}
=
\widetilde{\Barg}_{ijk}.
\qedhere
\]
\end{proof}

This proposition is the main elementary observation of the paper:
the triangular inconsistency factor of the phase comparison matrix associated with a family of qubit states is exactly the normalized Bargmann invariant of the corresponding triple.

\subsection{Consequences and first interpretation}

The preceding identity has several immediate consequences.

\begin{corollary}
For every non-orthogonal triple \((i,j,k)\), the triangular defect \(\kappa_{ijk}\) depends only on the rays
\[
[\psi_i],[\psi_j],[\psi_k]\in \CP^1.
\]
\end{corollary}

\begin{proof}
By Proposition~\ref{prop:kappa-bargmann}, \(\kappa_{ijk}\) is equal to the normalized Bargmann invariant, which is invariant under independent rephasing of the three representatives.
\end{proof}

\begin{corollary}
For every non-orthogonal triple \((i,j,k)\),
\[
\kappa_{ijk}=1
\qquad\Longleftrightarrow\qquad
\Barg_{ijk}\in \R_{+}^{\times}.
\]
\end{corollary}

\begin{proof}
By definition,
\[
\kappa_{ijk}=\frac{\Barg_{ijk}}{|\Barg_{ijk}|}.
\]
Hence \(\kappa_{ijk}=1\) if and only if \(\Barg_{ijk}\) has argument \(0\) modulo \(2\pi\), that is, if and only if \(\Barg_{ijk}\) is a positive real number.
\end{proof}

So, in this setting, multiplicative coherence of the phase comparison data on a triple is equivalent to the vanishing of the corresponding Bargmann phase.

\begin{remark}
The quantity \(\kappa_{ijk}\) should therefore not be interpreted merely as a formal defect.
It carries genuine quantum-geometric information.
Its argument
\[
\arg(\kappa_{ijk})
=
\arg(\Barg_{ijk})
\]
is the phase naturally associated with the triple \(([\psi_i],[\psi_j],[\psi_k])\).
In the next section, we recall how this phase is related to geometric phase and to the geometry of the Bloch sphere.
\end{remark}

\subsection{A first non-unitary remark}

Although \(\kappa_{ijk}\) is phase-valued, it should be remembered that it comes from the full transition-amplitude data
\[
g_{ij}=\ip{\psi_i}{\psi_j},
\]
not only from their phases.

Indeed, the Bargmann invariant splits naturally into modulus and phase:
\[
\Barg_{ijk}
=
|g_{ij}|\;|g_{jk}|\;|g_{ki}|\;\kappa_{ijk}.
\]
Thus the full triple datum consists of:
\begin{enumerate}[label=\textup{(\roman*)}]
\item a positive amplitude factor
\[
|g_{ij}|\;|g_{jk}|\;|g_{ki}|,
\]
which depends on transition probabilities;

\item a phase factor
\[
\kappa_{ijk}\in U(1),
\]
which measures the corresponding triangular defect.
\end{enumerate}

For this reason, the \(U(1)\)-valued comparison structure should not be seen as replacing the amplitude data, but rather as extracting one distinguished part of it.
The comparison-theoretic viewpoint is therefore naturally compatible with both unitary and non-unitary aspects of the quantum transition structure.

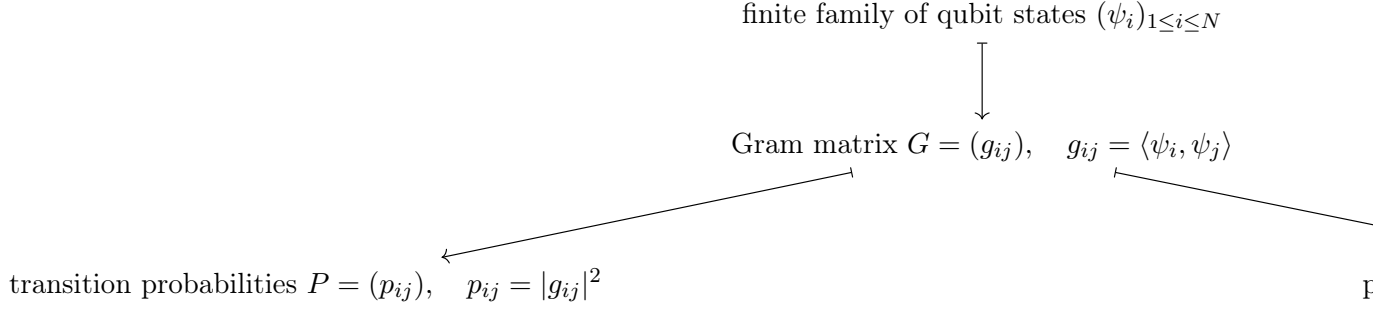
\begin{center}
\[
\begin{tikzcd}[column sep=large,row sep=large]
&
\text{finite family of qubit states } (\psi_i)_{1\le i\le N}
\arrow[d, mapsto]
&
\\
&
\text{Gram matrix } G=(g_{ij}),\quad g_{ij}=\langle \psi_i,\psi_j\rangle
\arrow[dl, mapsto]
\arrow[dr, mapsto]
&
\\
\text{transition probabilities } P=(p_{ij}),\quad p_{ij}=|g_{ij}|^2
&&
\text{phase comparison matrix } U=(u_{ij}),\quad
u_{ij}=\dfrac{g_{ij}}{|g_{ij}|}
\arrow[d, mapsto]
\\
&&
\text{triangular defects } \kappa_{ijk}=u_{ij}u_{jk}u_{ki}
\arrow[d, leftrightarrow]
\\
&&
\text{normalized Bargmann invariants } \widetilde{\Barg}_{ijk}
\arrow[d, mapsto]
\\
&&
\text{geometric phases / Pancharatnam phases}
\end{tikzcd}
\]
\end{center}
\section{Geometric phase and the Bloch sphere}

\subsection{Pancharatnam phase of a triple}

Let
\[
\psi_i,\psi_j,\psi_k\in \C^2
\]
be normalized qubit representatives such that all pairwise overlaps are nonzero.
As recalled above, the associated Bargmann invariant is
\[
\Barg_{ijk}
=
\ip{\psi_i}{\psi_j}\ip{\psi_j}{\psi_k}\ip{\psi_k}{\psi_i}.
\]

\begin{definition}
The \emph{Pancharatnam phase} of the triple \(([\psi_i],[\psi_j],[\psi_k])\) is
\[
\gamma_{ijk}:=\arg(\Barg_{ijk})\in \mathbb{R}/2\pi\mathbb{Z}.
\]
\end{definition}

By Proposition~\ref{prop:kappa-bargmann}, this phase is exactly the argument of the triangular defect.

\begin{proposition}
For every non-orthogonal triple \((i,j,k)\),
\[
\gamma_{ijk}=\arg(\kappa_{ijk}) \qquad \text{in } \mathbb{R}/2\pi\mathbb{Z}.
\]
\end{proposition}

\begin{proof}
By Proposition~\ref{prop:kappa-bargmann},
\[
\kappa_{ijk}=\frac{\Barg_{ijk}}{|\Barg_{ijk}|}.
\]
Taking the argument of both sides gives
\[
\arg(\kappa_{ijk})=\arg(\Barg_{ijk})=\gamma_{ijk}
\qquad \text{mod }2\pi.
\qedhere
\]
\end{proof}

Thus the local inconsistency factor of the \(U(1)\)-valued pairwise comparison structure is precisely the Pancharatnam phase of the corresponding triple of rays.
This is the first sense in which pairwise-comparison defects acquire a genuinely quantum-geometric meaning \cite{Pancharatnam1956,Berry1984,SamuelBhandari1988,MukundaSimon1993a,MukundaSimon2003}.

\subsection{Bloch sphere representation}

Every pure qubit state may be represented by a rank-one orthogonal projector
\[
\rho_\psi=\ket{\psi}\bra{\psi}.
\]
Equivalently, there exists a unique vector
\[
n_\psi\in \Sphere^2\subset \R^3
\]
such that
\[
\rho_\psi=\frac12\bigl(I+n_\psi\cdot \bm{\sigma}\bigr),
\]
where \(\bm{\sigma}=(\sigma_1,\sigma_2,\sigma_3)\) denotes the Pauli matrices.
The map
\[
[\psi]\longmapsto n_\psi
\]
identifies \(\CP^1\) with the Bloch sphere.

Let us write
\[
n_i:=n_{\psi_i}\in \Sphere^2
\]
for the Bloch vector associated with \([\psi_i]\).

\begin{proposition}\label{prop:transition-prob-bloch}
For every normalized qubit states \(\psi_i,\psi_j\),
\[
|\ip{\psi_i}{\psi_j}|^2
=
\Tr(\rho_i\rho_j)
=
\frac{1+n_i\cdot n_j}{2}.
\]
\end{proposition}

\begin{proof}
Using
\[
\rho_i=\frac12(I+n_i\cdot\bm{\sigma}),
\qquad
\rho_j=\frac12(I+n_j\cdot\bm{\sigma}),
\]
together with
\[
\Tr(I)=2,
\qquad
\Tr(\sigma_\alpha)=0,
\qquad
\Tr(\sigma_\alpha\sigma_\beta)=2\delta_{\alpha\beta},
\]
we obtain
\[
\Tr(\rho_i\rho_j)
=
\frac14\Tr\!\bigl((I+n_i\cdot\bm{\sigma})(I+n_j\cdot\bm{\sigma})\bigr)
=
\frac14(2+2\,n_i\cdot n_j)
=
\frac{1+n_i\cdot n_j}{2}.
\]
Since \(\rho_i=\ket{\psi_i}\bra{\psi_i}\) and \(\rho_j=\ket{\psi_j}\bra{\psi_j}\),
\[
\Tr(\rho_i\rho_j)=|\ip{\psi_i}{\psi_j}|^2.
\qedhere
\]
\end{proof}

This formula shows that transition probabilities between qubit states are encoded directly by the Euclidean geometry of the Bloch sphere.

\subsection{Bargmann invariants in Bloch coordinates}

The triple product has an equally natural expression in terms of Bloch vectors.

\begin{proposition}\label{prop:bargmann-bloch}
Let \(\psi_i,\psi_j,\psi_k\) be normalized qubit states, with associated Bloch vectors \(n_i,n_j,n_k\in \Sphere^2\).
Then
\[
\Barg_{ijk}
=
\Tr(\rho_i\rho_j\rho_k)
=
\frac14\Bigl(
1+n_i\cdot n_j+n_j\cdot n_k+n_k\cdot n_i
+i\, n_i\cdot (n_j\times n_k)
\Bigr).
\]
\end{proposition}

\begin{proof}
Since \(\rho_\ell=\ket{\psi_\ell}\bra{\psi_\ell}\), one has
\[
\Tr(\rho_i\rho_j\rho_k)
=
\ip{\psi_i}{\psi_j}\ip{\psi_j}{\psi_k}\ip{\psi_k}{\psi_i}
=
\Barg_{ijk}.
\]
It remains to compute the trace in Bloch coordinates.
Using
\[
\rho_\ell=\frac12(I+n_\ell\cdot\bm{\sigma}),
\]
we obtain
\[
\Tr(\rho_i\rho_j\rho_k)
=
\frac18\Tr\!\Bigl(
(I+n_i\cdot\bm{\sigma})
(I+n_j\cdot\bm{\sigma})
(I+n_k\cdot\bm{\sigma})
\Bigr).
\]
Now we use the standard Pauli identities
\[
\Tr(\sigma_\alpha)=0,
\qquad
\Tr(\sigma_\alpha\sigma_\beta)=2\delta_{\alpha\beta},
\qquad
\Tr(\sigma_\alpha\sigma_\beta\sigma_\gamma)=2i\,\varepsilon_{\alpha\beta\gamma}.
\]
Expanding the product and taking traces yields
\[
\Tr(\rho_i\rho_j\rho_k)
=
\frac18\Bigl(
2
+
2\,n_i\cdot n_j
+
2\,n_j\cdot n_k
+
2\,n_k\cdot n_i
+
2i\, n_i\cdot (n_j\times n_k)
\Bigr),
\]
hence
\[
\Tr(\rho_i\rho_j\rho_k)
=
\frac14\Bigl(
1+n_i\cdot n_j+n_j\cdot n_k+n_k\cdot n_i
+i\, n_i\cdot (n_j\times n_k)
\Bigr).
\]
This proves the claim.
\end{proof}

\begin{corollary}\label{cor:arg-bargmann-bloch}
For every non-orthogonal triple \((i,j,k)\),
\[
\arg(\kappa_{ijk})
=
\arg(\Barg_{ijk})
=
\arg\!\Bigl(
1+n_i\cdot n_j+n_j\cdot n_k+n_k\cdot n_i
+i\, n_i\cdot (n_j\times n_k)
\Bigr).
\]
\end{corollary}

\begin{proof}
By Proposition~\ref{prop:bargmann-bloch}, the factor \(1/4\) is real and positive, so it does not affect the argument.
The conclusion then follows from Proposition~\ref{prop:kappa-bargmann}.
\end{proof}

This formula already shows that the phase of the triangular defect is determined by the oriented geometry of the triple \((n_i,n_j,n_k)\) on the Bloch sphere.

\subsection{Geodesic triangles and geometric phase}

Let \(n_i,n_j,n_k\in \Sphere^2\) be pairwise non-antipodal points.
They determine a geodesic triangle on the Bloch sphere, whose oriented solid angle we denote by
\[
\Omega_{ijk}.
\]

The relation between Bargmann invariants and spherical geometry is classical: up to sign convention, the argument of the Bargmann invariant is equal to one half of the oriented solid angle of the corresponding geodesic triangle \cite{Pancharatnam1956,SamuelBhandari1988,MukundaSimon1993a,MukundaSimon2003,BengtssonZyczkowski2017}.

\begin{theorem}\label{thm:solid-angle}
With the standard Pancharatnam--Berry convention,
\[
\arg(\kappa_{ijk})
=
\arg(\Barg_{ijk})
=
-\frac{\Omega_{ijk}}{2}
\qquad \text{mod }2\pi.
\]
\end{theorem}

\begin{remark}
The sign depends on the orientation convention chosen for the geodesic triangle and for the geometric phase.
Some authors write \(+\Omega_{ijk}/2\) instead of \(-\Omega_{ijk}/2\).
For the purposes of the present paper, the important point is that the triangular defect is exactly the exponential of one half of the oriented solid angle.
\end{remark}

\begin{corollary}
For every non-orthogonal triple \((i,j,k)\),
\[
\kappa_{ijk}
=
\exp\!\left(-\frac{i}{2}\Omega_{ijk}\right)
\]
up to the chosen sign convention for \(\Omega_{ijk}\).
\end{corollary}

Thus the \(U(1)\)-valued inconsistency factor is not merely a formal defect: it is the holonomy-type phase associated with the corresponding spherical triangle.

\subsection{Interpretation}

The preceding discussion gives a simple geometric interpretation of pairwise-comparison defects in the qubit case.

At the probabilistic level, the quantities
\[
p_{ij}=|\ip{\psi_i}{\psi_j}|^2
\]
record transition probabilities and are determined by pairwise scalar products of Bloch vectors.
At the phase level, the quantities
\[
u_{ij}=\frac{\ip{\psi_i}{\psi_j}}{|\ip{\psi_i}{\psi_j}|}
\]
define a reciprocal \(U(1)\)-valued comparison matrix.
Its triangular defects
\[
\kappa_{ijk}=u_{ij}u_{jk}u_{ki}
\]
are exactly normalized Bargmann invariants, hence geometric phases.

In other words, a local inconsistency-type quantity in pairwise comparison theory becomes, for qubit states, a projective invariant of quantum geometry.
This is the basic bridge developed in the present note.

\begin{remark}
At this level, the picture remains essentially kinematic.
No dynamics, Hamiltonian evolution, or measurement protocol is required.
The relation between triangular defects and geometric phase is already present at the level of finite families of rays in \(\CP^1\).
\end{remark}

\section{Realizability, incomplete comparisons, and non-unitary aspects}

\subsection{Realizability through Gram matrices}

The preceding sections show that a non-orthogonal family of qubit states gives rise to a \(U(1)\)-valued reciprocal pairwise comparison matrix.
A natural converse question is the following: when does a prescribed set of pairwise transition data come from a family of qubit states?

At the level of full amplitudes, the answer is classical and is governed by Gram matrices.

\begin{proposition}\label{prop:gram-characterization}
Let
\[
G=(g_{ij})_{1\leq i,j\leq N}\in M_N(\C).
\]
Then \(G\) is the Gram matrix of a family of normalized qubit vectors
\[
\psi_1,\dots,\psi_N\in \C^2
\]
if and only if the following conditions hold:
\begin{enumerate}[label=\textup{(\roman*)}]
\item \(G\) is Hermitian;
\item \(G\) is positive semidefinite;
\item \(g_{ii}=1\) for all \(i\);
\item \(\rank(G)\leq 2\).
\end{enumerate}
\end{proposition}

\begin{proof}
If \(G\) is the Gram matrix of a family \((\psi_i)\subset \C^2\), then properties \textup{(i)}--\textup{(iv)} were established in Section~2.

Conversely, assume that \textup{(i)}--\textup{(iv)} hold.
Since \(G\) is positive semidefinite of rank at most \(2\), there exists a matrix
\[
V\in M_{2\times N}(\C)
\]
such that
\[
G=V^\ast V.
\]
Let \(v_i\in \C^2\) be the \(i\)-th column of \(V\).
Then
\[
g_{ij}=\ip{v_i}{v_j}
\]
for all \(i,j\).
Moreover,
\[
\|v_i\|^2=g_{ii}=1,
\]
so the vectors \(v_i\) are normalized.
Setting \(\psi_i:=v_i\) gives the required family of normalized qubit vectors.
\end{proof}

This immediately yields a realizability criterion for phase comparison matrices.

\begin{corollary}\label{cor:phase-realizability}
A matrix
\[
U=(u_{ij})\in \PC_N(U(1))
\]
is realizable by a non-orthogonal family of pure qubit states if and only if there exists a Hermitian positive semidefinite matrix
\[
G=(g_{ij})\in M_N(\C)
\]
such that:
\begin{enumerate}[label=\textup{(\roman*)}]
\item \(g_{ii}=1\) for all \(i\);
\item \(\rank(G)\leq 2\);
\item \(g_{ij}\neq 0\) for all \(i,j\);
\item
\[
\frac{g_{ij}}{|g_{ij}|}=u_{ij}
\qquad
\text{for all } i,j.
\]
\end{enumerate}
\end{corollary}

\begin{proof}
If \(U\) comes from a non-orthogonal qubit family, take \(G\) to be its Gram matrix.
Conversely, if such a matrix \(G\) exists, Proposition~\ref{prop:gram-characterization} provides a family of normalized qubit vectors with Gram matrix \(G\), and by construction the corresponding phase comparison matrix is exactly \(U\).
\end{proof}

\begin{remark}
Corollary~\ref{cor:phase-realizability} shows that the phase comparison structure alone is not arbitrary.
It must be compatible with positivity and with the rank-two constraint imposed by the qubit Hilbert space.
In particular, not every \(U(1)\)-valued reciprocal comparison matrix can arise from qubit states.
\end{remark}

A particularly simple realizable class is given by coherent phase comparisons.

\begin{proposition}\label{prop:coherent-realizable}
Let
\[
U=(u_{ij})\in \PC_N(U(1))
\]
be such that
\[
u_{ij}u_{jk}u_{ki}=1
\qquad
\text{for all } i,j,k.
\]
Then \(U\) is realizable by a family of normalized qubit vectors.
\end{proposition}

\begin{proof}
The coherence condition implies that there exist phases \(\lambda_1,\dots,\lambda_N\in U(1)\) such that
\[
u_{ij}=\lambda_i\lambda_j^{-1}.
\]
Choose any normalized vector \(\psi\in \C^2\), and define
\[
\psi_i:=\lambda_i\psi.
\]
Then
\[
\ip{\psi_i}{\psi_j}
=
\overline{\lambda_i}\lambda_j\,\ip{\psi}{\psi}
=
\overline{\lambda_i}\lambda_j
=
\lambda_i^{-1}\lambda_j.
\]
Hence
\[
\frac{\ip{\psi_i}{\psi_j}}{|\ip{\psi_i}{\psi_j}|}
=
\lambda_i^{-1}\lambda_j
=
u_{ji}.
\]
Replacing \(\lambda_i\) by \(\lambda_i^{-1}\) if needed, one gets exactly the prescribed matrix \(U\).
Equivalently, one may define the Gram matrix directly by
\[
g_{ij}:=\lambda_i\lambda_j^{-1}.
\]
This matrix has rank \(1\), is Hermitian positive semidefinite, and has diagonal entries equal to \(1\).
The conclusion then follows from Proposition~\ref{prop:gram-characterization}.
\end{proof}

Thus complete coherence corresponds to the most degenerate realizable case, namely a single projective ray with different phase representatives.

\subsection{Incomplete pairwise comparisons}

In the quantum setting, incompleteness appears naturally as soon as some transition amplitudes vanish.
Indeed, if
\[
\ip{\psi_i}{\psi_j}=0,
\]
then the normalized phase
\[
u_{ij}=\frac{\ip{\psi_i}{\psi_j}}{|\ip{\psi_i}{\psi_j}|}
\]
is not defined.
This leads naturally to a partially defined comparison structure.

\begin{definition}
Let \(\psi_1,\dots,\psi_N\in \C^2\) be normalized qubit vectors, and let
\[
g_{ij}:=\ip{\psi_i}{\psi_j}.
\]
The associated \emph{support graph} is the graph \(\Gamma_\psi\) with vertex set
\[
\{1,\dots,N\}
\]
and edge set
\[
E_\psi:=\bigl\{\{i,j\}\mid i\neq j,\ g_{ij}\neq 0\bigr\}.
\]
\end{definition}

Thus the missing edges correspond exactly to orthogonal pairs of states.

\begin{definition}
The \emph{partial phase comparison matrix} associated with the family \((\psi_i)\) is the partially defined matrix
\[
U^\partial=(u_{ij}),
\qquad
u_{ij}:=\frac{g_{ij}}{|g_{ij}|}
\]
defined precisely for those pairs \((i,j)\) such that \(g_{ij}\neq 0\).
\end{definition}

For every available pair, one still has
\[
u_{ji}=u_{ij}^{-1},
\qquad
u_{ii}=1.
\]
So \(U^\partial\) is an incomplete \(U(1)\)-valued reciprocal comparison matrix in the usual sense of incomplete pairwise comparison theory \cite{Harker1987,FedrizziGiove2007,BozokiFulopRonyai2010}.

\begin{definition}
Let \(i,j,k\) be pairwise distinct vertices such that
\[
g_{ij}\neq 0,\qquad g_{jk}\neq 0,\qquad g_{ki}\neq 0.
\]
Equivalently, assume that \((i,j,k)\) spans a triangle in the support graph \(\Gamma_\psi\).
Then the associated \emph{partial triangular defect} is
\[
\kappa_{ijk}^{\partial}:=u_{ij}u_{jk}u_{ki}\in U(1).
\]
\end{definition}

So triangular defects remain meaningful exactly on those triples for which all three pairwise comparisons are available.

\begin{proposition}
For every triple \((i,j,k)\) spanning a triangle in \(\Gamma_\psi\),
\[
\kappa_{ijk}^{\partial}
=
\frac{\ip{\psi_i}{\psi_j}\ip{\psi_j}{\psi_k}\ip{\psi_k}{\psi_i}}
{|\ip{\psi_i}{\psi_j}\ip{\psi_j}{\psi_k}\ip{\psi_k}{\psi_i}|}.
\]
In other words, the partial triangular defect is again the normalized Bargmann invariant of the triple.
\end{proposition}

\begin{proof}
This is the same computation as in Proposition~\ref{prop:kappa-bargmann}, restricted to triples for which all relevant overlaps are nonzero.
\end{proof}

\begin{remark}
The incomplete setting is not a pathological variant of the theory.
On the contrary, it is the natural framework as soon as orthogonality is allowed.
What survives from the complete theory is the local part of the phase-comparison structure, namely the reciprocal data on the support graph and the triangular defects on available triangles.
\end{remark}

There is also a qubit-specific geometric feature.

\begin{proposition}
Assume that the rays
\[
[\psi_1],\dots,[\psi_N]\in \CP^1
\]
are pairwise distinct.
Then every vertex of the complement graph of \(\Gamma_\psi\) has degree at most \(1\).
Equivalently, the orthogonality graph is a matching.
\end{proposition}

\begin{proof}
Fix a ray \([\psi]\in \CP^1\).
In \(\C^2\), its orthogonal complement is one-dimensional, hence determines a unique orthogonal ray.
Therefore, if \([\psi_i]\) is orthogonal to both \([\psi_j]\) and \([\psi_k]\), then \([\psi_j]=[\psi_k]\).
Under the assumption that the rays are pairwise distinct, this implies \(j=k\).
Hence each vertex is orthogonal to at most one other vertex.
\end{proof}

So, for distinct qubit rays, incompleteness is highly constrained: missing comparisons cannot form arbitrary patterns.

\subsection{Beyond the unitary phase level}

The phase comparison matrix \(U=(u_{ij})\) captures only the phase part of the transition amplitudes.
However, the full quantum data also involve the moduli
\[
|g_{ij}|=|\ip{\psi_i}{\psi_j}|.
\]
This becomes apparent already at the level of triples.

Indeed, for every non-orthogonal triple one has
\[
\Barg_{ijk}
=
g_{ij}g_{jk}g_{ki}
=
|g_{ij}|\,|g_{jk}|\,|g_{ki}|\,\kappa_{ijk}.
\]
Thus the full triple datum splits into:
\begin{enumerate}[label=\textup{(\roman*)}]
\item a positive amplitude factor
\[
|g_{ij}|\,|g_{jk}|\,|g_{ki}|,
\]
which depends only on transition probabilities;

\item a phase factor
\[
\kappa_{ijk}\in U(1),
\]
which is the normalized Bargmann invariant.
\end{enumerate}

\begin{remark}
From this viewpoint, the \(U(1)\)-valued pairwise comparison structure should be seen as an extraction from the full transition-amplitude structure, not as a replacement for it.
The phase comparison formalism isolates a particularly rigid and geometrically meaningful part of the quantum data, but it does not exhaust them.
\end{remark}

This also suggests possible extensions beyond pure states and unitary kinematics.
For mixed states or non-unitary evolutions, one may still consider transition probabilities, fidelities, and generalized geometric phases \cite{Uhlmann1976,Jozsa1994,SjoqvistEtAl2000,Sjoqvist2004,TongSjoqvistKwekOh2004,deFariaPizaNemes2003}.
In such settings, the pairwise-comparison viewpoint survives, but the simple \(U(1)\)-valued formalism must typically be replaced by a more elaborate structure.

\begin{remark}
The present note does not attempt to develop this generalization.
Its purpose is simply to emphasize that even in the elementary qubit setting, pairwise comparisons naturally occur at three distinct levels:
\[
\text{amplitudes} \quad\longrightarrow\quad \text{probabilities} \quad\longrightarrow\quad \text{phases}.
\]
The phase level is the one most directly related to reciprocal comparison matrices, while the amplitude level retains the non-unitary information encoded in the moduli.
\end{remark}

\begin{center}
\begin{tabular}{|c|c|c|c|}
\hline
Quantum datum & Formula & Pairwise-comparison viewpoint & Constraint
\\
\hline
Transition amplitude
&
\(g_{ij}=\langle \psi_i,\psi_j\rangle\)
&
complex comparison datum
&
Hermitian positive Gram matrix
\\
\hline
Transition probability
&
\(p_{ij}=|g_{ij}|^2\)
&
symmetric comparison strength
&
\(0\le p_{ij}\le 1\)
\\
\hline
Phase comparison
&
\(u_{ij}=\dfrac{g_{ij}}{|g_{ij}|}\)
&
\(U(1)\)-valued reciprocal PC matrix
&
defined only if \(g_{ij}\neq 0\)
\\
\hline
Triangular defect
&
\(\kappa_{ijk}=u_{ij}u_{jk}u_{ki}\)
&
local inconsistency factor
&
defined on available triangles
\\
\hline
Bargmann invariant
&
\(\Barg_{ijk}=g_{ij}g_{jk}g_{ki}\)
&
amplitude refinement of the defect
&
depends on moduli and phase
\\
\hline
Geometric phase
&
\(\arg(\Barg_{ijk})\)
&
phase of inconsistency
&
equals half the oriented solid angle
\\
\hline
\end{tabular}
\end{center}

\section{Concluding remarks}

In this note, we have proposed a simple pairwise-comparison reading of finite families of qubit states.
Starting from the transition amplitudes
\[
g_{ij}=\ip{\psi_i}{\psi_j},
\]
we distinguished three associated levels of pairwise data:
the full complex amplitudes \(g_{ij}\), the transition probabilities
\[
p_{ij}=|g_{ij}|^2,
\]
and, in the non-orthogonal case, the phase comparison coefficients
\[
u_{ij}=\frac{g_{ij}}{|g_{ij}|}\in U(1).
\]

The phase level is the one most directly related to reciprocal pairwise comparison theory.
It leads naturally to the triangular defects
\[
\kappa_{ijk}=u_{ij}u_{jk}u_{ki},
\]
which play the role of local inconsistency factors.
The main observation of the paper is that these defects are exactly the normalized Bargmann invariants of the corresponding triples of qubit states.
Their arguments are therefore geometric phases, and in the qubit case they are directly related to the oriented geometry of geodesic triangles on the Bloch sphere.

This provides a simple bridge between two languages that are usually kept apart.
On the one hand, pairwise comparison theory emphasizes reciprocity, coherence, and local defect measures.
On the other hand, elementary quantum geometry emphasizes transition amplitudes, projective invariance, and geometric phase.
In the qubit case, these two viewpoints meet in a particularly transparent way.

We also stressed that not every reciprocal \(U(1)\)-valued comparison matrix comes from qubit states.
Realizability is constrained by the existence of a Hermitian positive semidefinite Gram matrix of rank at most \(2\).
In addition, incomplete pairwise comparisons arise naturally as soon as orthogonality is present, so that the incomplete setting is not artificial but intrinsic to the quantum picture.

The present note is intentionally modest.
Its purpose is not to provide a full reconstruction theory, nor a general framework for quantum pairwise comparisons in arbitrary dimension.
Rather, it isolates a basic dictionary:
\[
\text{transition amplitudes}
\;\longrightarrow\;
\text{phase comparisons}
\;\longrightarrow\;
\text{triangular defects}
\;\longrightarrow\;
\text{geometric phases}.
\]
We hope that this elementary correspondence may be useful for further work in several directions, including higher-dimensional Hilbert spaces, mixed states, non-unitary evolutions, and more general group-valued comparison structures.

\vskip 12pt

\paragraph{\bf Data availability statement} No data is available for this work.

\vskip 12pt

\paragraph{\bf Conflict of interest statement} The author declares no conflict of interest.

\vskip 12pt

\paragraph{\bf Funding} No funding supported this work.

\vskip 12pt

\paragraph{\bf Acknowledgements} J.-P.M thanks the France 2030 framework programme Centre Henri Lebesgue ANR-11-LABX-0020-01 
for creating an attractive mathematical environment.

\vskip 12pt

\paragraph{\bf Author's Note on AI Assistance}
Portions of the text were developed with the assistance of a generative language model (OpenAI ChatGPT, based on the GPT-4 architecture). The AI was used to assist with drafting, editing, and standardizing the bibliography format. All mathematical content, structure, and theoretical constructions were provided, verified, and curated by the author. The author assumes full responsibility for the correctness, originality, and scholarly integrity of the final manuscript.

\end{document}